\documentclass[11pt, letterpaper, fleqn]{article}
\usepackage{amsmath, amssymb, amsthm, mathtools, tikz, todonotes, bbm, hyperref, fullpage}
\usepackage[noblocks]{authblk}

\newtheorem{theorem}{Theorem}[section]
\newtheorem{corollary}[theorem]{Corollary}
\newtheorem{lemma}[theorem]{Lemma}

\newtheorem{definition}[theorem]{Definition}
\newtheorem{claim}[theorem]{Claim}

\newtheorem{question}{Question}

\newcommand{\bits}{\{-1, 1\}}
\newcommand{\AND}{\textnormal{AND}}
\newcommand{\sr}{\mathrm{sr}}
\newcommand{\OR}{\textnormal{OR}}
\newcommand{\UPPcc}{\mathbf{UPP}^{\text{cc}}}
\newcommand{\UPP}{\mathbf{UPP}}
\newcommand{\PPcc}{\mathbf{PP}^{\text{cc}}}
\newcommand{\PP}{\mathbf{PP}}
\newcommand{\UPPdt}{\mathbf{UPP}^{\text{dt}}}

\newcommand{\MAJ}{\text{MAJ}}

\newcommand{\R}{\mathbb{R}}

\newcommand{\eps}{\varepsilon}
\newcommand{\poly}{\operatorname{poly}}
\newcommand{\sgn}{\operatorname{sgn}}

\newcommand{\bra}[1]{\{#1\}}
\newcommand{\pmone}{\bra{-1, 1}}

\renewcommand{\deg}{\mathrm{deg}}

\newcommand{\mathify}[1]{\ifmmode{#1}\else\mbox{$#1$}\fi}
\newcommand{\abs}[1]{\mathify{\left| #1 \right|}}

\usepackage{cite}
\title{Sign-Rank Can Increase Under Intersection}

\date{}

\begin{document}

\author[1]{Mark Bun}
\author[2]{Nikhil S.~Mande}
\author[2]{Justin Thaler}
\affil[1]{Simons Institute for the Theory of Computing and Boston University}
\affil[2]{Georgetown University}
\maketitle

\begin{abstract}
The communication class $\UPPcc$ is a communication analog of the Turing Machine complexity class $\PP$. It is characterized by a matrix-analytic complexity measure called sign-rank (also called dimension complexity), and is essentially the most powerful communication class against which we know how to prove lower bounds.

For a communication problem $f$, let $f \wedge f$ denote the function that evaluates $f$ on two disjoint inputs and outputs the $\AND$ of the results.  We exhibit a communication problem $f$ with $\UPPcc(f)= O(\log n)$, and $\UPPcc(f \wedge f) = \Theta(\log^2 n)$. 
This is the first result showing that $\UPP$ communication complexity can increase by more than a constant factor under intersection. We view this as a first step toward showing that $\UPPcc$, the class of problems with polylogarithmic-cost $\UPP$ communication protocols, is not closed under intersection.

Our result shows that the function class consisting of intersections of two majorities on $n$ bits has dimension complexity $n^{\Omega(\log n)}$. This matches an upper bound of  (Klivans, O'Donnell, and Servedio, FOCS 2002), who used it to give a quasipolynomial time algorithm for PAC learning intersections of polylogarithmically many majorities. Hence, fundamentally new techniques will be needed to learn this class of functions in polynomial time. 
\end{abstract}

\section{Introduction}\label{s: intro}
The unbounded-error communication complexity model $\UPPcc$ was
introduced by Paturi and Simon \cite{PS86} as a natural communication
analog of the Turing Machine complexity class $\PP$. 
In a $\UPPcc$ communication
protocol for a Boolean function $f(x, y)$, there are two parties, one with input $x$ 
and one with input $y$. The two parties engage in a private-coin randomized communication protocol,
at the end of which they are required to output $f(x,y)$ with probability strictly greater than $1/2$.
The cost of the protocol is the number of bits exchanged by the two parties. 
As is standard, we use the notation $\UPPcc$
not only to denote the communication model, but also the class of functions
solvable in the model by protocols of cost polylogarithmic in the size of the input.

Observe that success probability $1/2$ can be achieved with no communication by random guessing,
so the  $\UPPcc$ model merely requires a strict improvement over this trivial solution. 
Owing to this liberal acceptance criterion, 
$\UPPcc$ is a very powerful communication model, essentially the most powerful one against which we know
how to prove lower bounds. In particular, $\UPPcc$ is powerful enough to simulate
many other models of computing, and this makes $\UPPcc$ \emph{lower bounds} highly useful. As one example,
any function $f(x, y)$ computable by a Threshold-of-Majority circuit of size $s$
has $\UPPcc$ complexity at most $O(\log s)$, and this connection has been used to translate $\UPPcc$ lower bounds
into state of the art lower bounds against threshold circuits (see, for example,~\cite{Forster02, RS10, CM18, SheW19, BT16}).

$\UPPcc$ also happens to be characterized by a natural matrix-analytic complexity measure
called \emph{sign-rank} \cite{PS86}. Here,
the sign-rank of a matrix $M \in \{-1, 1\}^{N \times N}$  is the minimum rank of a real matrix whose entries agree in sign with $M$.  Equivalently,
$\sr(M) := \min_A\textnormal{rk}(A)$, where the minimum is over all matrices $A$ such that $A_{i, j} \cdot M_{i, j} > 0$ for all $i, j \in [N]$. 
Paturi and Simon~\cite{PS86} showed the following tight connection between $\UPPcc$ and sign-rank: if we associate a function $f(x, y)$ with the matrix $M=[f(x, y)]_{x, y}$, then 
the $\UPPcc$ communication complexity of $f$ equals $\log(\sr(M)) \pm \Theta(1)$.

While lower bounds on $\UPPcc$ complexity (equivalently, sign-rank) are useful
in complexity theory, \emph{upper bounds} on these quantities imply state of the art learning algorithms, including the
fastest known algorithms for PAC learning DNFs and read-once formulas \cite{ksdnf, ambainis2010any}. More specifically,
suppose we want to learn a concept class $\mathcal{C}$ of functions mapping $\{-1, 1\}^n$ to $\{-1, 1\}$.
$\mathcal{C}$ is naturally associated with a $|\mathcal{C}| \times 2^n$ matrix $M$,
whose $i$th row equals the truth table of the $i$th function in $\mathcal{C}$.
Then $\mathcal{C}$ can be distribution-independently PAC learned in time polynomial in the sign-rank of $M$. (The sign-rank
of $M$ is often referred to in the learning theory literature as the \emph{dimension complexity} of $\mathcal{C}$.)
Moreover, the resulting learning algorithm is robust to random classification noise, a property not satisfied
by the handful of known PAC learning algorithms that are \emph{not} based on dimension complexity. 

For the purpose of our work, one particularly important application
of the dimension-complexity approach to PAC learning was derived by Klivans et al.~\cite{KOS04}, who showed that the concept class consisting 
of intersections of 2 majority functions has dimension complexity at most ${\binom{n}{O(\log n)}} \leq n^{O(\log n)}$.
They thereby obtained a quasipolynomial time algorithm for PAC learning intersections of two majority functions.\footnote{
In fact, their algorithm runs in quasipolynomial time for intersections of polylogarithmic many majorities.} 
Prior to our work, it was consistent with current knowledge that the dimension complexity of this concept class is in
fact $\poly(n)$, which would yield a polynomial time PAC learning algorithm for intersections of constantly many majority functions.

\subsection{Our Results}
Despite considerable effort, progress on understanding sign-rank (equivalently, $\UPPcc$) has been slow.
Our lack of knowledge is highlighted via the following well-known open question (cf.~G\"{o}\"{o}s et al.~\cite{GPW18}).
Throughout, for any function $f \colon \bits^n \to \bits$, $f \wedge f$ denotes the function on twice as many inputs obtained by
evaluating $f$ on two disjoint inputs and outputting $-1$ only if both copies of $f$ evaluate to $-1$, i.e., $\left(f \wedge f\right)(x_1, x_2):=
f(x_1) \wedge f(x_2)$.

\begin{question} \label{firstquestion} 
Is the class $\UPPcc$ closed under intersection? In other words, suppose the function
$f(x, y) \colon \bits^{n} \times \bits^{n} \to \bits$ satisfies $\UPPcc(f) = O( \left(\log n\right)^c)$ for some constant $c$.
Is there always some constant $c_1$ (which may depend on $c$) such that $\UPP(f \wedge f) \leq O\left(\left(\log n\right)^{c_1}\right)$?
More generally and informally, if $\UPPcc(f)$ is ``small'', does this imply any non-trivial upper bound on $\UPPcc(f \wedge f)$?
\end{question}

\vspace{-1.5mm}
Prior to our work, essentially nothing was known about Question \ref{firstquestion}.
In particular, we are not aware of prior work ruling out the possibility that $\UPPcc(f \wedge f) \leq O(\UPPcc(f))$. 
On the other hand, for reasons that will become apparent in Section \ref{s:techniques},
there is good reason to suspect that there exists a function $f$ with $\UPPcc(f)=O(\log n)$, yet $\UPPcc(f \wedge f) \geq \Omega(n)$.
While we do not obtain a full resolution of Question \ref{firstquestion}, we do show for the first time
that $\UPPcc$ complexity can increase significantly under intersection. 

Babai, Frankl and Simon~\cite{BFS86} observed that there are two natural communication complexity analogs of the Turing machine class $\PP$, namely $\PPcc$ and $\UPPcc$.  It is well known~\cite{BRS95} that $\PPcc$ \emph{is} closed under intersection. Our work can be viewed as a first step towards showing that, in contrast, $\UPPcc$ is \emph{not}  closed under intersection.

\begin{theorem} \label{intromain}
There is a function $f(x, y) \colon \bits^{n} \times \bits^{n} \to \bits$ such that $\UPPcc(f) = O(\log n)$,
yet $\UPPcc(f \wedge f) = \Theta(\log^2 n)$. 
\end{theorem}

In fact, for each fixed $x \in \bits^n$, the function $f(x, y)$
from Theorem \ref{intromain} simply outputs the majority of some subset of the bits of $y$.
This yields the following corollary. 

\begin{corollary} \label{cor: learn}
Let $\mathcal{C}$ be the concept class in which each concept is the intersection of two majorities on $n$ bits. 
Then $\mathcal{C}$ has dimension complexity $n^{\Theta(\log n)}$.
\end{corollary}

Corollary~\ref{cor: learn} shows that the dimension complexity upper bound of Klivans et al.~\cite{KOS04} is tight for intersections 
of two majorities, and new approaches will be needed to PAC learn this concept class in polynomial time.
For context, we remark that learning intersections of majorities is a special case of the more general problem of 
learning intersections of many \emph{halfspaces}.\footnote{A halfspace is any function of the form $\sgn\left(\sum_{i=1}^n w_i \cdot x_i + w_0\right)$
for some real numbers $w_0, \dots, w_n$.}
The latter is a central and well-studied challenge in learning theory, as intersections of halfspaces are powerful enough to represent any convex set,
and they contain many basic problems (like learning DNFs) as special cases.
In contrast to the well-understood problem of learning a single halfspace, for which many efficient algorithms are known,
no $2^{o(n)}$-time algorithm is known for PAC learning  even the intersection of two halfspaces. There have been considerable 
efforts devoted to showing that learning intersections of halfspaces is a hard problem \cite{cryptohardness, feldman2006new,
khot2011hardness, bhattacharyya2017hardness}, but these results apply only to intersections
of many halfspaces, or make assumptions about the form of the output hypothesis of the learner. Our work can be seen
as a new form of evidence that learning intersections of even two majorities is hard.

\vspace{-2mm}
\subsection{Our Techniques}
\label{s:techniques}
$\UPPcc$ has a \emph{query complexity} analog, denoted $\UPPdt$ and defined as follows.
A $\UPPdt$ algorithm is a randomized algorithm which on input $x$, queries bits of $x$, and must output $f(x)$
with probability strictly greater than $1/2$; the cost of the protocol is the number of bits of $x$ queried.  
How $\UPPdt$ behaves under intersection is now well understood. More specifically,
it is known~\cite{Sherstov13b} that 
there is a function $f \colon \bits^n \to \bits$ (in fact, a halfspace) such that $\UPPdt(f)=O(1)$,
yet $\UPPdt(f \wedge f)= \Theta(n)$. Define the Majority function, which we denote by $\MAJ$, to be $-1$ if at least half of its input bits are $-1$. It is also known~\cite{Sherstov13} that $\MAJ$ satisfies 
$\UPPdt(\MAJ)=O(1)$,
yet $\UPPdt\left(\MAJ \wedge \MAJ\right)= \Theta(\log n)$.
Our goal in this paper is, to the extent possible, to show that the $\UPPcc$ communication model behaves similarly to
its query complexity analog. 

Over the course of the last decade, there has been considerable
progress in proving \emph{lifting theorems}~\cite{Sherstov11, GPW17, GKPW17}. These theorems seek
to show that if a function $f$ has large complexity in some query model $\mathbf{C}$,
then for some ``sufficiently complicated'' function $g$ on a ``small'' number of inputs, the composition $f \circ g$
has large complexity in the associated communication model (ideally, $\mathbf{C}^{\text{cc}}(f \circ g) \gtrsim \mathbf{C}^{dt}(f)$). 

Unfortunately, a ``generic'' lifting theorem for $\UPP$ complexity is not known.
That is, it is not know how to take an arbitrary function $f$ with high $\UPPdt$
complexity, and by composing it with a function $g$ on a small number of inputs,
yield a function with high $\UPPcc$ complexity.

However, as we now explain, some significant partial results have been shown in this direction.
It is well-known that $\UPPdt(f)$ is equivalent to an 
approximation-theoretic notion called \emph{threshold degree}, denoted $\deg_{\pm}(f)$ (see Appendix \ref{app:tdeg} for the definition).
The threshold degree of $f$ can in turn be expressed as the value of a 
certain (exponentially large) linear program. Linear programming duality
then implies that one can prove lower bounds on $\deg_{\pm}(f)$
by exhibiting good solutions to the dual linear program. We refer to such dual solutions as \emph{dual witnesses} for threshold degree. Sherstov 
\cite{sherstovsymm} and Razborov and Sherstov~\cite{RS10} showed that if $\deg_{\pm}(f)$
is large, and moreover this can be exhibited by a dual witness satisfying a certain \emph{smoothness}
condition, then there is a function $g$ defined on a constant number of inputs such that $f \circ g$ \emph{does} have large 
$\UPPcc$ complexity. Several recent works~\cite{BT16, BCHTV17, BT18, SheW19} have managed to prove new $\UPPcc$ lower bounds
by constructing, for various functions $f$, smooth dual witnesses exhibiting the fact that $\deg_{\pm}(f)$ is large.

Our key technical contribution is to bring this approach to bear on the function $F(x, y) = \MAJ(x) \wedge \MAJ(y)$.
Specifically, we show that the (known) threshold degree lower bound $\deg_{\pm}(F) \geq \Omega(\log n)$ can be 
exhibited by a smooth dual witness. 

We do this as follows.
Sherstov~\cite{Sherstov13} showed that for any function $f \colon \bits^n \to \bits$, the threshold degree of the function $F = f \wedge f$ is characterized
by the \emph{rational} approximate degree of $f$, i.e., the least total degree of real polynomials $p$ and $q$ such that $|f(x)-p(x)/q(x)| \leq 1/3$
for all $x \in \bits^n$. He then showed that the rational approximate degree of $\MAJ$ is $\Omega(\log n)$, thereby
concluding that $F(x, y)$ has threshold degree $\Omega(\log n)$.  \label{s:ratdeg}

From Sherstov's arguments, one can derive a dual witness $\psi$ for the fact that the rational approximate degree of $\MAJ$ is $\Omega(\log n)$,
and then transform $\psi$ into a dual witness $\phi$ for the fact that $F(x, y)$ has threshold degree $\Omega(\log n)$. Unfortunately,
neither $\psi$ nor $\phi$ satisfies the type of smoothness condition required by Razborov and Sherstov's machinery to yield $\UPPcc$ lower bounds.

The smoothness condition required for the Razborov-Sherstov machinery to work essentially states that the the mass of the dual witness $\psi$ has to be ``relatively large'' (a reasonably large fraction of what mass the uniform distribution would have placed) on a ``large'' set of inputs (the fraction of inputs which do not have large mass has to be small).

To construct a \emph{smooth} dual witness $\psi'$ for $F$, our primary technical contribution 
is to construct a \emph{smooth} dual witness $\phi'$ for the fact that the rational approximate degree of $\MAJ$ is $\Omega(\log n)$.
We then apply a different transformation, due to Sherstov~\cite{Sherstov14}, of $\phi'$ into a dual witness for the fact that the threshold degree of $F$ is $\Omega(\log n)$,
and we show that this transformation \emph{preserves the smoothness} of $\psi'$.

In a nutshell, our smooth dual witness for $\MAJ$ is obtained in two steps: first we define for all inputs $x$ whose Hamming weight lies in $[n/2- \lfloor n^{2/3} \rfloor, n/2 + \lfloor n^{2/3} \rfloor]$, a dual witness $\phi'_x$ that places a large mass on $x$ and not too much mass on other points.
Next, we define the final dual witness $\phi'(x)$ to be a certain weighted average over $x$ of all the dual witnesses thus obtained.
The resulting mass on $\phi'(x)$ for each $x$ of Hamming weight in $[n/2 - \lfloor n^{2/3}\rfloor, n/2 + \lfloor n^{2/3}\rfloor]$ is large enough, and the fraction of inputs whose Hamming weight is not in $[n/2 -\lfloor n^{2/3}\rfloor, n/2 +\lfloor n^{2/3}\rfloor]$ is small enough, to allow us to use the Razborov-Sherstov framework (Theorem~\ref{thm: sthrdegtosr}) to prove the desired sign-rank lower bound on the pattern matrix of $F$.

\section{Preliminaries}

\label{s:prelims}
All logarithms in this paper are taken base 2.  We use the notation $\exp(x)$ to denote $e^x$, where $e$ is Euler's number.  Given any finite set $X$ and any functions $f, g : X \to \R$, define $\|f\|_1 := \sum_{x \in X}|f(x)|$ and $\langle f, g \rangle := \sum_{x \in X}f(x)g(x)$. We refer to $\|f\|_1$ as the $\ell_1$-norm of $f$.  For any $x \in \pmone^{n}$, we use the notation $|x|$ to denote the \emph{Hamming weight} of $x$, which is the number of $-1$'s in the string $x$.

Paturi and Simon~\cite{PS86} showed the following equivalence between the sign-rank of a matrix and the $\UPPcc$ cost of its corresponding communication game.

\begin{theorem}\label{thm: ps}
For any $F : \bra{-1, 1}^{2n} \!\times\! \bra{-1, 1}^n\! \to \! \bra{-1, 1}$,
let $M_F$ denote its \emph{communication matrix}, defined by $M_F(x, y) = F(x, y)$.  Then,
$\UPPcc(F) = \log \sr(M_F) \pm O(1)$.
\end{theorem}

Let $n, N$ be positive integers such that $n$ divides $N$. Partition the set $[N] := \{1, \dots, N\}$ into $n$ disjoint blocks $\bra{1, 2, \dots, N/n}, \bra{N/n + 1, \dots, 2N/n}, \dots, \bra{(n-1)N/n + 1, \dots, N}$.  Define the set $\mathcal{P}(N, n)$ to be the collection of subsets of $[N]$ which contain exactly one element from each block.  For $x \in \pmone^n$ and $S \in \mathcal{P}(N, n)$, let $x|_S = (x_{s_1}, \dots, x_{s_n})$, where $s_1 < s_2 < \cdots < s_n$ are the elements of $S$.
\begin{definition}[Pattern matrix]\label{defn: pm}
For any function $\phi : \pmone^n \to \R$, the $(N, n, \phi)$-pattern matrix $M$ is defined as follows.
\[
M = [\phi(x|_S) \oplus w]_{x \in \pmone^{N}, (S, w) \in \mathcal{P}(N, n) \times \pmone^{n}}.
\]
Note that $M$ is a $2^N \times (N/n)^n 2^n$ matrix.
\end{definition}

In a breakthrough result, Forster~\cite{Forster02} proved that an upper bound on the spectral norm of a sign matrix implies a lower bound on its sign-rank.
Razborov and Sherstov~\cite{RS10} established a generalization of Forster's theorem~\cite{Forster02} that can be used to prove sign-rank lower bounds for pattern matrices.  
Specifically, we require the following result, implicit in their work~\cite[Theorem 1.1]{RS10}.
\begin{theorem}[Implicit in~\cite{RS10}]\label{thm: sthrdegtosr}
Let $f : \pmone^n \rightarrow \bra{-1, 1}$ be any Boolean function and $\alpha > 1$ be a real number.  Suppose there exists a function $\phi : \pmone^n \rightarrow \R$ satisfying the following conditions.
\begin{itemize}
    \item 
    $\sum_{x \in \pmone^n}|\phi(x)| = 1.$
    \item
    For all polynomials $p$ of degree at most $d$,
    $\sum_{x \in \pmone^n}\phi(x)p(x) = 0.$
    \item
    $f(x) \cdot \phi(x) \geq 0~\forall x \in \pmone^n.$
    \item $|\phi(x)| \geq \gamma$ for all but a $\Delta$ fraction of inputs $x \in \pmone^n$.
\end{itemize}
Then, the sign-rank of the $(N, n, f)$-pattern matrix $M$ can be bounded below as 
\[
\sr(M) \geq \frac{\gamma}{\frac{1}{2^n}\left(\frac{n}{N}\right)^{d/2} + \gamma\Delta}.
\]
\end{theorem}

We require the following well-known combinatorial identity.

\begin{claim}\label{claim:combinatorial-identity}
For every polynomial $p$ of degree less than $2n$, we have $\sum_{t = -n}^n (-1)^t \binom{2n}{n + t} p(t) = 0$.
\end{claim}

Recall from Section \ref{s:techniques} that the rational $\epsilon$-approximate degree of $f$
is the least degree of two polynomials $p$ and $q$ such that $|f(x) - p(x)/q(x)| \leq \epsilon$ for all $x$ 
in the domain of $f$. 
 Sherstov \cite[Theorem 6.9]{Sherstov14} showed that a dual witness to the rational approximate degree of any function $f$
 can be converted to a threshold degree dual witness for $\OR_n \circ f$.  
 Implicit in his theorem is the fact that a \emph{smooth} dual witness to the rational approximate degree of $f$ can be converted to a \emph{smooth} dual witness for the threshold degree of $\OR_n \circ f$.  More precisely, the following result is established by the proof of~\cite[Theorem 6.9]{Sherstov14}.\footnote{
 In Theorem \ref{thm: srdegtosthrdeg}, the functions $\psi_1$ and 
 $\psi_0$ together form a dual witness for the fact that the rational $\delta$-approximate degree of $f$ is at least $d$, while
 $\Psi$ is a dual witness to the fact that $\deg_{\pm}(F) \geq d$. See Appendix \ref{app} for details.
 However, we will not exploit this interpretation of Theorem \ref{thm: srdegtosthrdeg} in our analysis.}

\begin{theorem}[Sherstov \cite{Sherstov14}]\label{thm: srdegtosthrdeg}
Let $f : \pmone^n \to \pmone$ be any function.  Let $F$ denote $\OR_t \circ f : \pmone^{nt} \to \pmone$, and $\delta > \eps > 0$ be any real numbers.

Suppose there exist functions $\psi_0, \psi_1 \colon \{-1, 1\}^n \to \R$ that are not identically 0 and satisfy the following properties:
\begin{align}
f(x) = 1 & \implies \psi_0(x) \geq \delta|\psi_1(x)|, \label{property1} \\
f(x) = -1 & \implies \psi_1(x) \geq \delta|\psi_0(x)|, \label{property2}\\
\deg(p) < d & \implies \langle \psi_0, p \rangle = 0 \text{ and } \langle \psi_1, p \rangle = 0. \label{property3}
\end{align}
Then there exist functions $A, B : \bra{-1, 1}^{nt} \to \R$ such that $\Psi = \frac{1}{\delta}A - \frac{1}{\eps}B$ satisfies the following properties.
\begin{align}
\deg(p) \leq \min\bra{\lfloor \eps^2 t \rfloor d, d} & \implies \langle \Psi, p \rangle = 0. \label{thm7conc1}\\
F(x) \cdot \Psi(x_1, \dots, x_t) & \geq (\delta - \eps)^{2t} \prod_{i = 1}^t \abs{\psi_0(x_i)} \text{ for all } x \in \pmone^{nt}. \label{thm7conc2}\\
|A(x_1, \dots, x_t)| & \leq \prod_{i = 1}^t|\psi_0(x_i)|  \text{ for all }  x = (x_1, \dots, x_t) \in \pmone^{nt}. \label{thm7conc3}\\
|B(x_1, \dots, x_t)| & \leq \prod_{i : f(x_i) = 0}\abs{\psi_0(x_i)} \cdot \prod_{i : f(x_i) = 1}\delta \psi_1(x_i) + \prod_{i = 1}^t \left(\abs{\psi_0(x_i)} - \delta\psi_1(x_i)\right) \notag\\
&  \text{for all }  x = (x_1, \dots, x_t) \in \pmone^{nt}. \label{thm7conc4}
\end{align}
\end{theorem}

\section{A Smooth Dual Witness for Majority}

Our main technical contribution in this paper is captured in Theorem \ref{thm:smooth-dual} below. This theorem constructs a \emph{smooth} dual witness $R$ for the hardness of rationally approximating the sign function on $\{0, \pm 1, \dots, \pm n\}$ (cf. Appendix \ref{app} for details of this interpretation of Theorem \ref{thm:smooth-dual}). We defer the proof until Section \ref{s: pfmain}. 

\begin{theorem} \label{thm:smooth-dual}
Let $1 \le d \le \frac{1}{3}\log n$ and let $n$ be odd. There exists a function $R : \{0, \pm 1, \dots, \pm n\} \to \R$ such that
\begin{itemize}
\item $\displaystyle \sum_{t = -n}^{n} |R(t)| = 1.$
\hfill \refstepcounter{equation} \label{eqn:sgn-L1} \textup{(\theequation)}
\item For $\delta = \exp(-18/(n^{1/(6d)}))$ and every $t = 1, 2, \dots, n$,
\begin{equation} \label{eqn:sgn-approx-sym}
R(t) \ge \delta |R(-t)|.
\end{equation}
\item If $p : \{0, \pm 1, \dots, \pm n\} \to \R$ is any polynomial of degree less than $d - 2$, then
\begin{equation} \label{eqn:sgn-phd}
\langle R, p \rangle = 0.
\end{equation}
\item For every $t \in \{0, \pm 1, \pm 2, \dots, \pm \lfloor n^{2/3} \rfloor\}$ we have
\begin{equation} \label{eqn:sgn-smooth}
|R(t)| \ge \Omega \left(\frac{1}{n^{20}}\right).
\end{equation}
\end{itemize}
\end{theorem}

The following theorem shows how to convert the (univariate) function $R$ from Theorem \ref{thm:smooth-dual} into a dual witness for the (multivariate) $\MAJ$ function.

\begin{theorem}\label{thm: univtomultivphd}
Let $1 \le d \le \frac{1}{3}\log n$ and let $n$ be odd.
Let $R : \bra{0, \pm1, \dots, \pm n} \rightarrow \mathbb{R}$ be any function obtained in Theorem~\ref{thm:smooth-dual}.
Then, the multivariate polynomial $R' : \pmone^{2n} \rightarrow \mathbb{R}$ defined by $R'(x) = R(n - |x|)/\binom{2n}{|x|}$ 
satisfies the following properties.

\begin{itemize}
    \item $\displaystyle     \|R'\|_1 = 1.$
    \hfill \refstepcounter{equation} \label{eqn:maj-L1} \textup{(\theequation)}
    \item For $\delta = \exp(-18/(n^{1/(6d)}))$ and every $t = 1, 2, \dots, n$,
    \begin{equation} \label{eqn:maj-approx-sym}
    R'(x) \ge \delta |R'(y)|
    \end{equation}
    for any $x, y \in \pmone^{2n}$ such that $|x| = n - t, |y| = n + t$.
    \item For any polynomial $p$ of degree at most $d - 2$,
    \begin{equation}\label{eqn:maj-phd}
    \langle R', p \rangle = 0.
    \end{equation}
    \item For all $x \in \pmone^{2n}$ such that $n - \lfloor n^{2/3} \rfloor \leq |x| \leq n + \lfloor n^{2/3} \rfloor$,
    \begin{equation}\label{eqn:maj-smooth}
    |R'(x)| \geq \Omega\left(\frac{1}{n^{20} \cdot 2^{2n}}\right).
    \end{equation}
\end{itemize}
\end{theorem}
\begin{proof}
To establish Equation \eqref{eqn:maj-L1}, observe:
\begin{align*}
\|R'\|_1 & = \sum_{x \in \pmone^{2n}}|R'(x)| = \sum_{t = 0}^{2n}\left(\sum_{x \in \pmone^{2n} : |x| = t}|R'(x)|\right)\\
& = \sum_{t = 0}^{2n} \binom{2n}{t} |R(n - t)|/\binom{2n}{t} = \sum_{t = -n}^{n}|R(t)| = 1,
\end{align*}
where the last equality follows from Equation~\eqref{eqn:sgn-L1}.
Equation~\eqref{eqn:maj-approx-sym} follows directly from Equation~\eqref{eqn:sgn-approx-sym} and the definition of $R'$.

To establish Equation \eqref{eqn:maj-phd}, consider  any polynomial $p : \pmone^{2n} \rightarrow \R$ of degree at most $d - 2$.
For any permutation $\sigma \in S_{2n}$, define the polynomial $p_{\sigma}$ by $p_{\sigma}(x_1, \dots, x_{2n}) = p(x_{\sigma(1)}, \dots, x_{\sigma(2n)})$.  Note that, since $R'$ is symmetric, $\langle R', p_{\sigma}\rangle = \langle R', p \rangle$ for all $\sigma \in S_{2n}$.  Define $q = \mathbb{E}_{\sigma \in S_{2n}}[p_{\sigma}]$. Note that $q$ is symmetric and $\langle R', p \rangle = \langle R', q \rangle$.  It is a well-known fact~(cf.~\cite{MP69}) that $q$ can be written as a polynomial $q'$ of degree at most $d - 2$ in the variable $\sum_{i = 1}^{2n}x_i$, and so can $R'$.
Hence,
$
\langle R', p \rangle = \langle R', q \rangle = \sum_{t = 0}^{2n} \binom{2n}{t} \frac{R(n - t)}{\binom{2n}{t}} \cdot q'(t) = 0,
$
where the final equality holds by Equation~\eqref{eqn:sgn-phd}.

To establish Equation \eqref{eqn:maj-smooth}, observe that by Equation~\eqref{eqn:sgn-smooth} and the definition of $R'$, we have that for all $x \in \pmone^{2n}$ such that $|x| \in \left[n - \lfloor n^{2/3} \rfloor, n + \lfloor n^{2/3} \rfloor\right]$,
$
\abs{R'(x)} \geq \Omega\left(\frac{1}{n^{20} \cdot \binom{2n}{|x|}}\right) \geq \Omega\left(\frac{1}{n^{20} \cdot 2^{2n}}\right).
$
\end{proof}

We are ready to derive a lower bound on the sign-rank of the $(4n^2, 4n, \OR_2 \circ \MAJ_{2n})$-pattern matrix.

\begin{theorem}\label{thm: main}
The $(4n^2, 4n, \OR_2 \circ \MAJ_{2n})$-pattern matrix $M$ satisfies
$
\sr(M) \geq n^{\Omega(\log n)}.
$
\end{theorem}

\begin{proof}
Let $F$ denote the function $\OR_2 \circ \MAJ_{2n}$ in this proof.
Set $d = \log n / 100$ and consider the function $R : \bra{0, \pm1, \dots, \pm n} \to \mathbb{R}$ obtained via Theorem~\ref{thm:smooth-dual}. Define the function $R' : \bra{0, \pm1, \dots, \pm n} \to \mathbb{R}$ by $R'(t) = R(-t)$.
Define the functions $\psi_0, \psi_1 : \pmone^{2n} \to \mathbb{R}$ by $\psi_1(x) = R(n - |x|)/\binom{2n}{|x|}$, and $\psi_0(x) = R'(n - |x|)/\binom{2n}{|x|}$ .
We now verify that $\psi_0, \psi_1$ satisfy the conditions in Theorem~\ref{thm: srdegtosthrdeg} for $\delta = \exp(-18/(n^{1/(6d)})) = \exp(-18/n^{100/6 \log n}) = \exp(-18/2^{100/6}) > 0.99$.  Set $\eps = \delta \cdot c$, where $c > 0$ is a constant such that $0.99 > \delta \cdot c > 1/{\sqrt{2}}$.
\begin{itemize}
\item By the definitions of $\psi_0, \psi_1$ and Equation~\eqref{eqn:maj-approx-sym}, Properties \eqref{property1} and \eqref{property2} in the statement of Theorem~\ref{thm: srdegtosthrdeg} are satisfied.
\item Equation~\eqref{eqn:maj-phd} implies that $\langle \psi_0, p \rangle = \langle \psi_1, p \rangle = 0$ for any polynomial $p$ of degree at most $d - 2$,
and hence Property \eqref{property3} is satisfied. 
\end{itemize}
Moreover, Equation~\eqref{eqn:maj-smooth} implies that $|\psi_0(x)|, |\psi_1(x)| \geq \Omega\left( \frac{1}{n^{20} \cdot 2^{2n}}\right)$ for all $x \in \pmone^{2n}$ such that $n - \lfloor n^{2/3} \rfloor \leq |x| \leq n + \lfloor n^{2/3} \rfloor$, and Equation~\eqref{eqn:maj-L1} implies $\|\psi_0\|_1 = \|\psi_1\|_1 = 1$.
Theorem~\ref{thm: srdegtosthrdeg} now implies the existence of a function $\Psi$ satisfying the following properties.
\begin{itemize}
\item By Equation \eqref{thm7conc1}, $\deg(p) < \min\bra{\lfloor 2\eps^2 \rfloor \cdot ((\log n)/100 - 2), (\log n)/100 - 2} \implies \langle \Psi, p \rangle = 0$.
Since $\eps > 1/\sqrt{2}$, this implies that \[
\deg(p) < (\log n)/100 - 2 \implies \langle \Psi, p \rangle = 0.
\]
    \item By Equation \eqref{thm7conc2}, $\Psi(x) \cdot F(x) \geq 0$ for all $x \in \pmone^{2n} \times \pmone^{2n}$.
    \item We now note that the functions $A$ and $B$ obtained in Theorem~\ref{thm: srdegtosthrdeg} have $\ell_1$-norm at most a constant.  Since  $\|\psi_0\|_1 = \|\psi\|_1 = 1$, we use Equation~\eqref{thm7conc3} to conclude that
    \[
    \sum_{x_1, x_2 \in \pmone^{2n} \times \pmone^{2n}}|A(x_1, x_2)| \leq \sum_{x_1 \in \pmone^{2n}}|\psi_0(x_1)| \cdot \sum_{x_2 \in \pmone^{2n}}|\psi_0(x_2)| = 1.
    \]
    By Equation~\eqref{thm7conc4}, we have
    \begin{align*}
    \sum_{x_1, x_2 \in \pmone^{2n}}|B(x_1, x_2)| \leq \max\bra{\|\psi_0\|_1, \delta\|\psi_1\|_1}^2 + \|\psi_0\|_1^2 + \delta \|\psi_0\|_1 \|\psi_1\|_1 + \delta^2 \|\psi_1\|_1^2,
    \end{align*}
    which is at most a constant, since $\delta = O(1)$.
    
    Combined with the fact that $\eps$ is a constant, we conclude $\|\Psi\|_1 \leq \frac{1}{\delta}\|A\|_1 + \frac{1}{\eps}\|B\|_1 \leq O(1)$.
    \item By Equation \eqref{thm7conc2}, $F(x) \cdot \Psi(x_1, x_2) \geq (\delta - \eps)^{4} \abs{\psi_0(x_1)} \cdot \abs{\psi_0(x_2)}~\forall x \in \pmone^{4n}$.  This implies that for $|x_1|, |x_2| \in [n - \lfloor n^{2/3} \rfloor, n + \lfloor n^{2/3} \rfloor]$,
    \[
    |\Psi(x_1, x_2)| \geq \Omega\left(\frac{1}{n^{40} \cdot 2^{4n}}\right),
    \]
    since $\delta - \eps = \Omega(1)$
    \item By a standard Chernoff bound, the number of inputs in $\pmone^{2n} \times \pmone^{2n}$ such that $|x_1|, |x_2| \in [n - \lfloor n^{2/3} \rfloor, n + \lfloor n^{2/3} \rfloor]$ is at least $1 - 2\exp(-n^{1/3}/3)$. 
\end{itemize}

Plugging $f = \OR_2 \circ \MAJ_{2n}$ and $\phi = \frac{\Psi}{\|\Psi\|_1}$ into Theorem~\ref{thm: sthrdegtosr}, we conclude that the sign-rank of the $(4n^2, 4n, \OR_2 \circ \MAJ_{2n})$ pattern matrix $M$ is bounded below as
\[
\sr(M) \geq \Omega\left(\frac{\frac{1}{n^{40}} \cdot \frac{1}{2^{4n}}}{\left(\frac{1}{n^{(\log n / 200) - 1}} \cdot \frac{1}{2^{4n}}\right) + \left(\frac{1}{n^{40}} \cdot \frac{1}{2^{4n}} \cdot 2\exp(-n^{1/3}/3)\right)}\right) \geq n^{\Omega(\log n)}.
\]
\end{proof}

We are now ready to prove Theorem~\ref{intromain}.

\begin{proof}[Proof of Theorem~\ref{intromain}]
Note that the function $\AND \circ \MAJ(x) = \overline{\OR \circ \MAJ(\overline{x})}$.  Consider the dual witness $\phi = \frac{\Psi}{\|\Psi\|_1}$ obtained for the threshold degree of $\OR_2 \circ \MAJ_{2n}$ in the previous proof.  Note that the function $\phi'$ defined by $\phi'(x) = -\phi(\overline{x})$ acts as a dual witness for the threshold degree of $\AND_2 \circ \MAJ_{2n}$, and satisfies all the conditions in Theorem~\ref{thm: sthrdegtosr} with the same parameters as in the proof of Theorem~\ref{thm: main}.  Proceeding in exactly the same way as in the previous proof, we conclude that sign-rank of the 
$(4n^2, 4n, \AND_2 \circ \MAJ_{2n})$ pattern matrix $M'$ is bounded below as
\begin{equation}\label{eqn: srmain}
\sr(M') \geq n^{\Omega(\log n)}.
\end{equation}

Denote by $f$ the communication game corresponding to the $(2n^2, 2n, \MAJ_{2n})$ pattern matrix. 
For completeness, we now sketch a standard $\UPPcc$ protocol of cost $O(\log n)$ for $f$.  Note that Alice holds $2n^2$ input bits, and Bob holds a $(2n \cdot \log n)$-bit string indicating the ``relevant bits'' in each block of Alice's input and a $2n$-bit string $w$.
Bob sends Alice the index of a uniformly random relevant bit using $\log (2n^2)$ bits of communication. Alice responds with her value $b$ of that input bit, and Bob outputs $b \oplus w_i$.  It is easy to check that this is a valid $\UPPcc$ protocol, and it has cost $O(\log n)$.

One can verify by the definition of pattern matrices (Definition~\ref{defn: pm}) that the communication game corresponding to the $(4n^2, 4n, \AND_2 \circ \MAJ_{2n})$ pattern matrix $M'$ equals $f \wedge f$.  By Theorem~\ref{thm: ps} and Equation~\eqref{eqn: srmain}, we obtain that 
\[
\UPP(f \wedge f) = \Theta(\log \sr(M')) = \Omega(\log^2 n).
\]

As mentioned in Section~\ref{s: intro}, the result of Klivans et al.~\cite{KOS04} implies that $\sr(M') = O(\log^2 n)$.
Thus, the function $f$ satisfies $\UPPcc(f) = O(\log n)$, but $\UPPcc(f \wedge f) = \Theta(\log^2 n)$.
\end{proof}

Corollary~\ref{cor: learn} follows immediately from the previous proof and the definition of pattern matrices.

\section{Proof of Theorem~\ref{thm:smooth-dual}}
\label{s: pfmain}
The rest of this paper is dedicated towards proving Theorem~\ref{thm:smooth-dual}.
Before proving the theorem, we describe the main auxiliary construction and prove some preliminary facts about it.

Let $\Delta = \lfloor n^{1/(3d)} \rfloor \ge 2$. Fix any $u \in \{1, \dots, \lfloor n^{2/3} \rfloor  - 1, \lfloor n^{2/3} \rfloor \}$. Define the set 
$$S_u = \{\pm u, \pm u\Delta, \pm u\Delta^2, \dots, \pm u\Delta^{d-1}\}.$$ Define the polynomial $r_u : \bra{0, \pm 1, \dots, \pm n} \to \R$ by

\[r_u(t) = \frac{1}{(2n)!} \prod_{i = 0}^{d-1} \left(t - \left(u\Delta^i \sqrt{\Delta}\right)\right) \prod_{s \notin S_u} (t - s).\]

Since $n$ is odd, notice that $\sgn(r_u(t)) = (-1)^t$, for $t \in \{u, u\Delta, u\Delta^2, \dots, u\Delta^{d-1}\}$, and $r_u(t) = 0$ for $t \notin S_u$.

Define
\[p_u(t) = \binom{2n}{n+t} r_u(t) = \begin{cases}
(-1)^{n - t} \cdot \dfrac{\prod\limits_{i = 0}^{d-1} \left(t - \left(u\Delta^i \sqrt{\Delta}\right)\right)}{\prod\limits_{\substack{s \in S_u \\ s \ne t}} (t - s)} & \text{ if } t \in S_u \\
& \\
0 & \text{ otherwise.}
\end{cases}\]

The following claim tells us that for any $u \in \bra{1, \dots, \lfloor n^{2/3} \rfloor}$, the function $p_u$ places a reasonably large mass on  input $-u$.
\begin{claim} \label{claim:middle-big}
\[|p_u(-u)| \ge \frac{\sqrt{\Delta} + 1}{2} \cdot u^{-(d-1)} \cdot\Delta^{-(d-1)^2/2}.\]
\end{claim}

\begin{proof}
We calculate
\begin{align*}
|p_u(-u)| &= \frac{u(\sqrt{\Delta} + 1)}{2u} \cdot \prod_{i = 1}^{d - 1} \frac{u(\Delta^i \sqrt{\Delta} + 1)}{u^2(\Delta^i + 1) (\Delta^i - 1)}&\tag{pairing terms corresponding to $u\Delta^i$ and $-u\Delta^i$}\\
& = \frac{\sqrt{\Delta} +  1}{2} \cdot u^{-(d-1)} \cdot \prod_{i = 1}^{d - 1}\frac{\Delta^{i + \frac{1}{2}} + 1}{\Delta^{2i} - 1} \ge \frac{\sqrt{\Delta} +  1}{2} \cdot u^{-(d-1)}\cdot \Delta^{(d - 1)/2} \cdot \prod_{i = 1}^{d - 1}\frac{\Delta^{i}}{\Delta^{2i}} \\
& = \frac{\sqrt{\Delta} + 1}{2} \cdot u^{-(d-1)} \cdot\Delta^{-(d-1)^2/2}.
\end{align*}

\end{proof}

The next claim tells us that the mass placed by $p_u$ on other points in its support is small.
\begin{claim} \label{claim:tails-small}
For every $j = 1, 2, \dots, d-1$,
\[|p_u(-u\Delta^j)| \le e^{4} \cdot \Delta^{-(j^2 - 3j - 2) / 2} \cdot \left(\frac{\sqrt{\Delta} + 1}{2}  \cdot u^{-(d-1)} \cdot \Delta^{-(d-1)^2/2}\right). \]
\end{claim}

\begin{proof}
We calculate
\begin{align*}
|p_u(-u\Delta^j)| &= \frac{u(\Delta^j\sqrt{\Delta} + \Delta^j)}{2u\Delta^j} \cdot \prod_{i = 0}^{j - 1} \frac{u(\Delta^i \sqrt{\Delta} + \Delta^j)}{u^2(\Delta^i + \Delta^j)(\Delta^{j} - \Delta^{i})} \cdot \prod_{i = j+1}^{d - 1} \frac{u(\Delta^i \sqrt{\Delta} + \Delta^j)}{u^2(\Delta^i + \Delta^j)(\Delta^{i} - \Delta^{j})} \tag{pairing terms corresponding to $u\Delta^i$ and $-u\Delta^i$}\\
& \le \frac{\sqrt{\Delta} +1}{2} \cdot u^{-(d-1)} \cdot \prod_{i = 0}^{j - 1} \frac{\sqrt{\Delta}}{\Delta^j - \Delta^{i}} \cdot \prod_{i = j+1}^{d - 1} \frac{\sqrt{\Delta}}{\Delta^i - \Delta^{j}}\\
&\le \frac{\sqrt{\Delta} +1}{2} \cdot (\sqrt{\Delta} \cdot u^{-1})^{d-1} \cdot \prod_{i = 0}^{j - 1} \frac{\Delta^{j-i} \cdot \Delta^{-j}}{\Delta^{j-i}-1} \cdot \prod_{i = j+1}^{d - 1} \frac{\Delta^{-i} \cdot \Delta^{i - j}}{\Delta^{i - j} - 1} \\
&\le \frac{\sqrt{\Delta} +1}{2} \cdot (\sqrt{\Delta} \cdot u^{-1})^{d-1} \cdot \prod_{i = 0}^{j - 1} \Delta^{-j} \cdot \prod_{i = j+1}^{d - 1} \Delta^{-i} \cdot \left(\prod_{k = 1}^\infty \frac{\Delta^k}{\Delta^k - 1}\right)^2 \\
&\le \frac{\sqrt{\Delta} + 1}{2}  \cdot (\sqrt{\Delta} \cdot u^{-1})^{d-1} \cdot \Delta^{-j^2 - (d(d-1) - (j+2)(j+1))/2} \cdot \exp \left(2 \sum_{k = 1}^\infty \frac{1}{\Delta^k - 1}\right) \tag{since $1 + x \leq e^x$ for all $x \in \mathbb{R}$}\\
& \le \frac{\sqrt{\Delta} + 1}{2} \cdot u^{-(d-1)} \cdot \Delta^{-(j^2 - 3j - 2) / 2} \cdot \Delta^{-(d-1)^2/2} \cdot \exp \left(4 \sum_{k = 1}^\infty \frac{1}{\Delta^k}\right) \tag{since $\Delta \geq 2$}\\
&\le e^{4} \cdot \frac{\sqrt{\Delta} + 1}{2} \cdot u^{-(d-1)} \cdot \Delta^{-(j^2 - 3j - 2) / 2} \cdot \Delta^{-(d-1)^2/2} \tag{again using $\Delta \geq 2$}.
\end{align*}

\end{proof}

The following claim tells us that for each $u$ and $j$, the masses placed by $r_u$ (and hence $p_u$) on $u \Delta^j$ and $-u \Delta^{j}$ are comparable.
\begin{claim} \label{claim:approx-symmetry}
For every $j = 0, 1, \dots, d-1$, we have \begin{align*}
|r_u(-u\Delta^j)| & \ge |r_u(u\Delta^j)| \ge \exp(-18/\sqrt{\Delta}) |r_u(-u\Delta^j)|\\
\text{and~} |p_u(-u\Delta^j)| & \ge |p_u(u\Delta^j)| \ge \exp(-18/\sqrt{\Delta}) |p_u(-u\Delta^j)|.
\end{align*}
\end{claim}

\begin{proof}
We may write the ratio
\[
\frac{|p_u(u\Delta^j)|}{|p_u(-u\Delta^j)|} = \frac{|r_u(u\Delta^j)|}{|r_u(-u\Delta^j)|} = \prod_{i = 0}^{j-1} \frac{u(\Delta^j - \Delta^i\sqrt{\Delta})}{u(\Delta^j + \Delta^i \sqrt{\Delta})} \cdot \prod_{i = j}^{d-1} \frac{u(\Delta^i\sqrt{\Delta} - \Delta^j)}{u(\Delta^j + \Delta^i \sqrt{\Delta})}.
\]
This is a product of terms smaller than $1$, yielding the first inequality. For the second, we follow Sherstov's argument~\cite[Theorem 5.3]{Sherstov13} and note that this product is at least
\begin{align*}
\left(\prod_{i =1}^\infty \frac{\Delta^{i/2} - 1}{\Delta^{i/2} + 1} \right)^2 &\ge \exp \left( -5 \sum_{i = 1}^\infty \frac{1}{\Delta^{i/2}} \right) \tag*{since $(a-1)/(a + 1) > \exp(-2.5/a)$ for $a \geq \sqrt{2}$}\\
& = \exp\left(\frac{-5}{\sqrt{\Delta}}\sum_{i = 0}^\infty \frac{1}{\Delta^{i/2}}\right) \geq \exp\left(\frac{-5}{\sqrt{\Delta}} \cdot \frac{1}{1 - 1/\sqrt{2}}\right)\tag*{since $\Delta \geq 2$}\\
&\ge \exp\left(-\frac{18}{\sqrt{\Delta}}\right).
\end{align*}
\end{proof}

Putting the three claims together, we obtain the following conclusion, which states that the mass placed by $p_u$ on $-u$ and $u$ is a relatively large fraction of its $\ell_1$-norm.

\begin{lemma} \label{lemma:two-point-smoothness}
$|p_u(-u)| \geq \|p_u\|_1/(8\Delta^{2}e^4)$ and
$|p_u(-u)| \ge |p_u(u)| \ge \frac{ \exp(-18/\sqrt{\Delta} - 4) }{8\Delta^{2}} \cdot \|p_u\|_1$.
\end{lemma}

\begin{proof}
We bound the ratio
\begin{align*}
\frac{\|p_u\|_1}{|p_u(-u)|} &\le 2\sum_{j = 0}^{d-1} \frac{|p_u(-u\Delta^j)|}{|p_u(-u)|} \tag*{ by the first inequality in Claim~\ref{claim:approx-symmetry}} \\
&\le 2\left( 1 +  \sum_{j = 0}^{d-1} e^{4} \Delta^{-(j^2 - 3j - 2) / 2}\right) \tag*{ by Claims~\ref{claim:middle-big} and \ref{claim:tails-small}} \\
& \leq 2 + 2 e^{4} \left(\sum_{j = 0}^{3} \Delta^{-(j^2 - 3j - 2) / 2} + \sum_{j = 4}^{\infty} \Delta^{-(j^2 - 3j - 2) / 2}\right)\\
&\le 8 \Delta^{2}\cdot e^{4} \cdot \sum_{k = 1}^\infty \Delta^{-k} \le 8 \cdot \Delta^{2} \cdot e^{4}. \tag*{since $\Delta \geq 2$}
\end{align*}
By the above and the second inequality in Claim~\ref{claim:approx-symmetry},
\[
|p_u(u)| \geq \exp(-18/\sqrt{\Delta})|p_u(-u)| \geq \frac{\exp(-18/\sqrt{\Delta} - 4)}{8\Delta^2} \cdot \|p_u\|_1.
\]
\end{proof}

We are now ready to prove Theorem~\ref{thm:smooth-dual}.

\begin{proof}[Proof of Theorem~\ref{thm:smooth-dual}]
Define the function
$P(t) = \sum_{u = 1}^{\lfloor n^{2/3} \rfloor} u^{20} \cdot \frac{p_u(t)}{\|p_u\|_1}.$
We claim that the function $R : \bra{0, \pm 1, \dots, \pm n} \to \pmone$ defined by 
$
R(t) = \frac{(-1)^tP(t)}{\|P\|_1}
$
satisfies the conditions in Theorem~\ref{thm:smooth-dual}.
\begin{itemize}
    \item Clearly, 
    $
    \sum_{t = -n}^n|R(t)| = 1,
    $ i.e., $R$ satisfies Equation \eqref{eqn:sgn-L1}.
    \item By Claim~\ref{claim:approx-symmetry}, for every $u = 1, \dots, \lfloor n^{2/3} \rfloor$ and every $t = 1, \dots, n$,~ $(-1)^t p_u(t) \ge \delta |p_u(-t)|$ for $\delta = \exp(-18/\sqrt{\Delta}) = \exp(-18/n^{(1/6d)})$. Therefore, for all such $t$ we also have $(-1)^t P(t) \ge \delta |P(t)|$, which implies $R(t) \geq \delta|R(-t)|$ for every $t = 1, 2, \dots, n$.
    \item We have
    \[
    R(t) = \frac{(-1)^tP(t)}{\|P\|_1} =\frac{(-1)^t}{\|P\|}\sum_{u = 1}^{\lfloor n^{2/3} \rfloor} u^{20} \cdot \frac{p_u(t)}{\|p_u\|_1} = \frac{(-1)^t}{\|P\|_1} \binom{2n}{n + t}\sum_{u = 1}^{\lfloor n^{2/3} \rfloor}u^{20} \cdot \frac{r_u(t)}{\|p_u\|_1}.
    \]
Since each $r_u$ is a polynomial of degree at most $(2n + 1) - d$, Claim~\ref{claim:combinatorial-identity} implies that for any polynomial $p$ of degree at most $d - 2$,
$
\langle R, p \rangle = 0.
$
\item It now remains to verify the smoothness condition. Fix a point $v \in \{1, \dots, \lfloor n^{2/3} \rfloor\}$. Since $\sgn(p_u(v)) = (-1)^v$ for all $u$ and for all $v > 0$, we have that
\begin{align*}
\frac{|P(v)|}{\|P\|_1} &\ge v^{20} \cdot \frac{ |p_v(v)| \cdot \|p_v\|_1^{-1}}{\sum_{u = 1}^{\lfloor n^{2/3} \rfloor}u^{20}} \ge \frac{\exp(-18/\sqrt{\Delta} - 4) / 8\Delta^{2}}{\lfloor n^{2/3} \rfloor \cdot (\lfloor n^{2/3} \rfloor)^{20}} \tag*{by Lemma~\ref{lemma:two-point-smoothness}}\\
& \geq \frac{\exp(-18/\sqrt{2} - 4)}{8n^{15}} \ge \frac{e^{-15}}{8n^{15}}. \tag*{since $n^{1/3} \geq \Delta = \lfloor n^{1/3d} \rfloor \geq 2$}
\end{align*}
If $v < 0$, the argument needs some more care because we do not have the guarantee that $\sgn(p_u(v)) = (-1)^v$.  The large mass placed by $p_{-v}$ on the point $v$ plays a crucial role.
\begin{align*}
|P(v)| &\ge \frac{(-v)^{20} \cdot |p_{-v}(v)|}{\|p_{-v}\|_1} - \sum_{\substack{u = 1 \\ u\ne v}}^{\lfloor n^{2/3} \rfloor} u^{20} \cdot \frac{p_u(-v)}{\|p_u\|_1} \\
&\ge \frac{(-v)^{20}}{8\Delta^2e^4} - \sum_{j = 1}^{\lfloor \log_{\Delta}(-v) \rfloor}(-v\Delta^{-j})^{20} \cdot \frac{p_{-v\Delta^{-j}}(v)}{\|p_{-v\Delta^{-j}}\|_1} \tag*{by Lemma~\ref{lemma:two-point-smoothness}, the definition of $p_u$ and its support}\\
& \geq (-v)^{20}\left[\frac{1}{8\Delta^2e^4} - e^4\sum_{j = 1}^\infty \Delta^{-20j} \cdot \Delta^{(-j^2 + 3j + 2)/2}\right] \tag*{ by Claims~\ref{claim:middle-big} and \ref{claim:tails-small}}\\
& = (-v)^{20}\left[\frac{1}{8\Delta^2e^4} - e^4\sum_{j = 1}^\infty \Delta^{(-j^2 - 37j + 2)/2}\right]
\geq (-v)^{20}\left[\frac{1}{8\Delta^2e^4} - e^4\sum_{j = 1}^\infty \Delta^{-18j}\right]\\
& \geq (-v)^{20}\left[\frac{1}{8\Delta^2e^4} -\frac{e^4}{\Delta^{17}}\right] = (-v)^{20}\left[\frac{\Delta^{15} - 8e^8}{8\Delta^{17}e^4}\right] 
\geq \frac{20}{\Delta^{17}} \geq \frac{20}{n^6}\tag*{since $n^{1/3} \geq \Delta \geq 2$ and $(-v) \geq 1$}
\end{align*}
Thus, we have that for $v < 0$,
\[
\frac{|P(v)|}{\|P\|_1} \geq \frac{20}{n^6{\sum_{u = 1}^{\lfloor n^{2/3} \rfloor}u^{20}}} \geq \frac{20}{n^6\lfloor n^{2/3} \rfloor \cdot (\lfloor n^{2/3} \rfloor)^{20}} \geq \frac{20}{ n^{20}}.
\]
\end{itemize}
\end{proof}

\section{Conclusion}
We have exhibited a communication problem $f$ with $\UPPcc(f)= O(\log n)$, and
$\UPPcc(f \wedge f) = \Theta(\log^2 n)$. 
This is the first result showing that $\UPP$ communication complexity
can increase by more than a constant factor under intersection.
As a consequence, we have concluded that the dimension-complexity-based quasipolynomial time PAC learning algorithm of \cite{KOS04}
for learning intersections of polylogarithmically many majorities is optimal. That is, 
new learning algorithms not based on dimension complexity will be required to learn this class in polynomial time.

A glaring open question left by our work is whether 
the class of problems with polylogarithmic $\UPPcc$
complexity is closed under intersection. Our results represent an important first step 
in this direction. 
It would also be very interesting to extend our result that dimension-complexity-based algorithms cannot PAC learn
intersections of two majorities in polynomial time, to rule out an even larger class of learning algorithms. 
Specifically, it would be very interesting to show that no algorithm working in the important \emph{statistical query} model
 \cite{sqmodel} can learn this concept class in polynomial time.

\bibliography{bibo}

\appendix

\section{Other Related Work on Sign Rank}
Alon, Frankl and R{\"{o}}dl~\cite{AFR85}
proved lower bounds on the sign-rank of random matrices. The first nontrivial lower bounds for
explicit matrix families was obtained in the breakthrough work of Forster \cite{Forster02}, who proved strong
lower bounds on the sign-rank of  any sign matrix
with small spectral norm. Several subsequent works improved and generalized Forster's method, and studied
the relationship of sign-rank to other important complexity measures in learning theory and circuit complexity
\cite{alon2014sign, forsteretal, forster2002smallest, linial2007complexity}. 
As previously mentioned, Sherstov \cite{sherstovsymm} and Razborov and Sherstov \cite{RS10} showed 
that a smooth dual witness for the fact that a function has large threshold degree implies
that a related function has large sign-rank. Razborov and Sherstov used this result to great effect,
constructing a smooth dual witnesses for the large threshold degree of a certain DNF, and thereby 
giving an AC$^0$ function with exponential sign-rank. This answered an old question of Babai, Frankl, and Simon \cite{BFS86}. 
Recent works have quantitatively strengthened these sign-rank lower bounds for AC$^0$ \cite{BT16, BT18, SheW19}.
Another work, due to Bouland et al.~\cite{BCHTV17}, managed to apply the methods to \emph{simpler}
functions within AC$^0$, and thereby resolved several old open questions about the power of statistical zero knowledge proofs. 

\section{Threshold Degree, Rational Approximate Degree, Sign-Rank and Duality}
\label{app}
In this appendix, we define threshold degree and introduce its dual formulation, 
and also introduce the dual formulation of rational approximate degree. We also
 explain the connection between the threshold degree, rational approximate degree, and the sign-rank
 of pattern matrices.
The material in this appendix is standard and provided only for intuition; none of it is actually required 
to prove the results in this paper. %Our presentation of this material owes a textual debt to \cite{BT18}.

\subsection{Threshold Degree and Its Dual Formulation}\label{app:tdeg}
The \emph{threshold degree} of a Boolean function $f\colon \{-1, 1\}^n \to \bits$, denoted $\deg_{\pm}(f)$, 
 is the least degree of a real polynomial $p$ that sign-represents $f$, i.e., $p(x) \cdot f(x) > 0$
for all $x \in \{-1, 1\}^n$.

To describe
	the dual formulation of threshold degree, we need to introduce some terminology.
	
	\begin{definition}
	Let $\psi \colon \{-1, 1\}^n \to \R$ be any real-valued function on the Boolean hypercube. 
	Recall (see Section \ref{s:prelims}) that, given another function $p \colon \{-1, 1\}^n \to \R$, we let $\langle \psi, p\rangle := \sum_{x \in \{-1, 1\}^n} \psi(x) \cdot p(x)$,
	and refer to $\langle \psi, p\rangle$ as the \emph{correlation} of $\psi$ with $p$. Also, $\|\psi\|_1 := \sum_{x \in \{-1, 1\}^n} |\psi(x)|$, and we refer to $\|\psi\|_1$ as the $\ell_1$-norm of $\psi$.
	If $\langle \psi, p\rangle = 0$ for all polynomials $p$ of degree at most $d$, we say that $\psi$ has \emph{pure high degree} at least $d$.
	\end{definition}

	The following standard theorem provides the aforementioned dual formulation of threshold degree.
	\begin{theorem}\label{theorem:tdegdual}
		
		Let $f:\{-1, 1\}^n \to \{-1,1\}$. Then $\deg_{\pm}(f) > d$ if and only if there is a real function $\psi: \{-1, 1\}^n \to \R$ such that:
		\begin{enumerate}
			\item (Pure high degree): $\psi$ has pure high degree at least $d$.
			\item (Non-triviality): $\|\psi\|_1 > 0$.
			\item (Sign Agreement): $\psi(x) \cdot f(x) \ge 0$ for all $x \in \{-1, 1\}^n$.
		\end{enumerate}
		$\psi$ is called a \emph{dual polynomial} for the fact that $\deg_{\pm}(f) > d$.
	\end{theorem}
	
	\medskip \noindent \textbf{The relationship between threshold degree and sign-rank of pattern matrices.}
	Suppose that $\deg_{\pm}(f) = d$. Then it is not hard to see that the $(N, n, f)$-pattern matrix $M$
	has sign-rank at most $\binom{N}{d} \leq N^d$. This is because each entry $M_{x, (S, w)}$ of $M$ can be written
	as the sign of a linear combination of monomials in $x$, where each monomial in the linear combination has degree at most $d$.
Theorem \ref{thm: sthrdegtosr} shows that this sign-rank upper bound is essentially tight,
	so long as there is a dual polynomial $\psi$ for that fact that $\deg_{\pm}(f) \geq d$, such that $\psi$ satisfies an extra smoothness condition.

\subsection{Rational Approximate Degree and Its Dual Formulation}
In Section \ref{s:ratdeg}, we defined the rational $\epsilon$-approximate degree of $f$ to be
the least total degree of real polynomials $p$ and $q$ such that $|f(x)-p(x)/q(x)| \leq \epsilon$
for all $x$ in the domain of $f$.
The functions $\psi_0$ and $\psi_1$ appearing in
Theorem \ref{thm: srdegtosthrdeg} constitute a dual witness to the large rational approximate
degree of $f$. However, to describe the duality theory of rational approximate degree,
it is helpful to introduce a related notion, due to Sherstov \cite{Sherstov14}.

\begin{definition}[Sherstov \cite{Sherstov14} Definition 6.4] \label{def: rfd0d1}
For $d_0, d_1 > 0$ and a function $f \colon \pmone^n \to \pmone$, define
$R(f, d_0, d_1)$ as the infimum over all $\epsilon > 0 $ for which there exist polynomials $p_0, p_1$ of
degree at most $d_0, d_1$, respectively, such that:
\begin{itemize}
\item $f(x) = 1 \Longrightarrow  |p_1(x)| < \epsilon \cdot p_0(x)$,
\item $f(x) = -1 \Longrightarrow |p_0(x)| < \epsilon \cdot p_1(x)$.
\end{itemize}
\end{definition}

\medskip
\noindent \textbf{The relationship between rational approximate degree and Definition~\ref{def: rfd0d1}.}
Sherstov \cite{Sherstov14} showed that, for constant $\epsilon > 0$, the rational $\epsilon$-approximate degree of $f$
is at most $O(d)$ if and only if there is some constant $c>0$ such that $R(f, c \cdot d, c \cdot d) \leq \epsilon$. 
And unlike rational approximate degree itself, $R(f, d_0, d_1)$ has a clean dual characterization.

\begin{theorem}[Sherstov \cite{Sherstov14} Theorem 6.6] \label{thm: rdegequivalence}
For $d_0, d_1 > 0$ and a function $f \colon \pmone^n \to \pmone$, $R(f, d_0, d_1) \geq \epsilon$ if and only if 
there exist $\psi_0, \psi_1 \colon \pmone^n \to \R$ such that 
\begin{itemize}
\item $\psi_0(x) \geq \epsilon |\psi_1(x)|$ for all $x \in f^{-1}(1)$.
\item $\psi_1(x) \geq \epsilon |\psi_0(x)|$ for all $x \in f^{-1}(-1)$.
\item $\psi_0$ has pure high degree at least $d_0$.
\item $\psi_1$ has pure high degree at least $d_1$.
\item $\|\psi_0\|_1, \|\psi_1\|_1 >0$.
\end{itemize}
\end{theorem}

\medskip \noindent \textbf{The Connection Between Rational Approximate Degree and Threshold Degree of Intersections.}
Let $F=f \wedge f$. Beigel et al.~\cite{BRS95} famously showed that if $f$
has rational approximate degree at most $d$, then $F$ has threshold degree at most $O(d)$
(see \cite[Section 1.2]{Sherstov14} for a lucid explanation of this upper bound).
The dual formulations of threshold degree  (Theorem \ref{theorem:tdegdual}) and rational approximate degree (Theorem \ref{thm: rdegequivalence}) together
allow us to interpret
Theorem \ref{thm: srdegtosthrdeg} as showing that the upper bound of Beigel et al.~is tight.
That is, Theorem \ref{thm: srdegtosthrdeg} shows that if $f$ has rational $\delta$-approximate degree at least $d$ for an appropriate value of $\delta$,
then $F$ has threshold degree at least $d$. In fact, Theorem \ref{thm: srdegtosthrdeg} explicitly translates
a dual witness to the high rational approximate degree of $f$ into a dual witness to the high threshold degree of $F$.

\end{document}